\documentclass{article}
\pdfoutput=1
\usepackage[pdftex]{}

\usepackage[utf8]{inputenc}
\usepackage{dtklogos} 

\usepackage{graphicx}
\usepackage{float}

\usepackage{amsfonts,amsmath,amstext,amsthm,amssymb} 

\newtheorem{koro}{Corollary}[section]

\title{How to extract data from proprietary software database systems using TCP/IP?\footnote{This copyright protected article is for scientific purpose only. According to German law this scientific document is not suposed to act or be abused against §202 StGB.}}
\author{Marc Burdon\footnote{Dipl.-Inform. Marc Burdon is a computer scientist, who graduated at the University of Bonn (Germany).} (burdon@cs.uni-bonn.de).\footnote{Copyright \copyright 2014 by Marc Burdon.}}
\date{December 2011 (publishing April 2014)\footnote{This paper was written in December 2011, but refreshed and pubished now, in April 2014.}}

\begin{document}

\maketitle

\section{Abstract}
This document is a {\bf white paper}\footnote{This paper (Version \today) does not give any scientific progress, but it shows potentials of the chosen topics on higher scientific level.} about how to connect reverse engineering and programing skills to extract data from a proprietary implementation of a database system to build EML-Tools\cite{EML} for data format conversion into raw data.\\
This article shows how to access data of a source software system without any interface for data conversion. We discuss how raw data can be transfered into structural format by using XML or any other custom designed software solution.
For demonstration purposes only, we will use a CRM\cite{CRM} system called {\em Harmony$^{\mbox{\tiny\textregistered}}$} by Harmony$^{\mbox{\tiny\textregistered}}$ Software AG\footnote{All trademarks are property of their owners, as Harmony$^{\mbox{\tiny\textregistered}}$ is of Harmony Software AG.}, the programing language {\em Python} and methods of computer security, which are used to get quick access to the raw data.

\section{Requirements}

\begin{itemize}
\item Python, programing skills in Python\footnote{Please refer to docs.python.org} and in network programing\cite{C[13]};
\item Reverse Engineering skills\cite{E[05]};
\item XML knowledge (optionally recommended)\footnote{Please refer to www.w3.org}\cite{XML};
\item Wireshark\footnote{Please refer to www.wireshark.org} and one trail copy of Harmony$^{\mbox{\tiny\textregistered}}$.
\end{itemize}

\section{Theory}
\subsection{Method to Access Data}
\begin{figure}[t]
 \centering 
\includegraphics[width=1.0\textwidth]{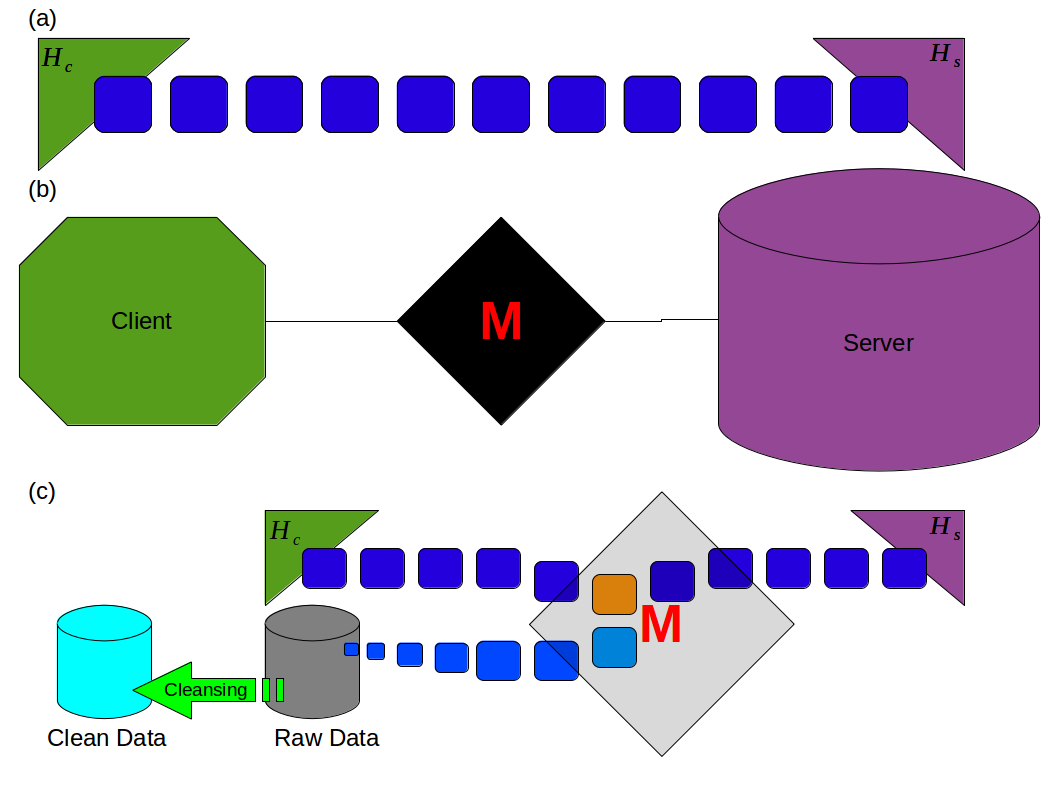} 
\caption[Scheme of concept]{ Figure (a) shows the regular data package streams between client and server. During Man-in-the-middle attack (b) data packages will be copied and stored as raw data (c). When data extraction is finished, the data is ready to be migrated to another system e.g. a SQL database.} 
\label{fig:man-in-middle_diagram.png} 
\end{figure}

In the beginning, we realize that offline data extraction from files on the file system of the database server will cost much more time than we will need using {\em TCP/IP} technology, because the reverse engineering process is much more quicker in this case, which is using the TCP/IP client of Harmony$^{\mbox{\tiny\textregistered}}$ CRM software solution. We have to know the structure of the Harmony$^{\mbox{\tiny\textregistered}}$ CRM software system that we just name $H$.
The relevant core of $H$ consists of a TCP/IP client and a TCP/IP server, which is the database server.
We name Harmony$^{\mbox{\tiny\textregistered}}$ Software client as $H_c$  and Harmony$^{\mbox{\tiny\textregistered}}$ software database server as $H_s$.
We define a data package $d=(b_0, \cdots, b_n)$ with $b_0, \cdots, b_n$, which are bytes coding a string character. We also define the set of data packages $D$ and a set of raw data $R$. We also use a self made network package filter $M(\cdot,\cdot)$. 
${ M(\cdot,\cdot)}$ is actually a function mapping data packages transfered by using TCP/IP from $H_c$ and $H_s$ (and vice versa) to $R$, so formally $M$, $M: D \times D \rightarrow R$, is commutative.
We use $M$ for a {\em Man-in-the-middle} attack\cite{ATT} on $H$. Technically $M$ is a hybrid TCP/IP of both client and server system, which is $H_c$ is connected to and $H_s$ is connected by. 
The data connection between $H_c \mbox{ and }H_s$ is not encrypted by default.
It can be encrypted by using e.g. {\em TLS/SSL} technology -- though the funtionality is implemented by using {\em SSH}\cite{ssh3}\cite{ssh0} --, but by default in local area networks traffic is not supposed to be encrypted by manufacturer.

\begin{figure}[h]
 \centering 
\includegraphics[width=1.0\textwidth]{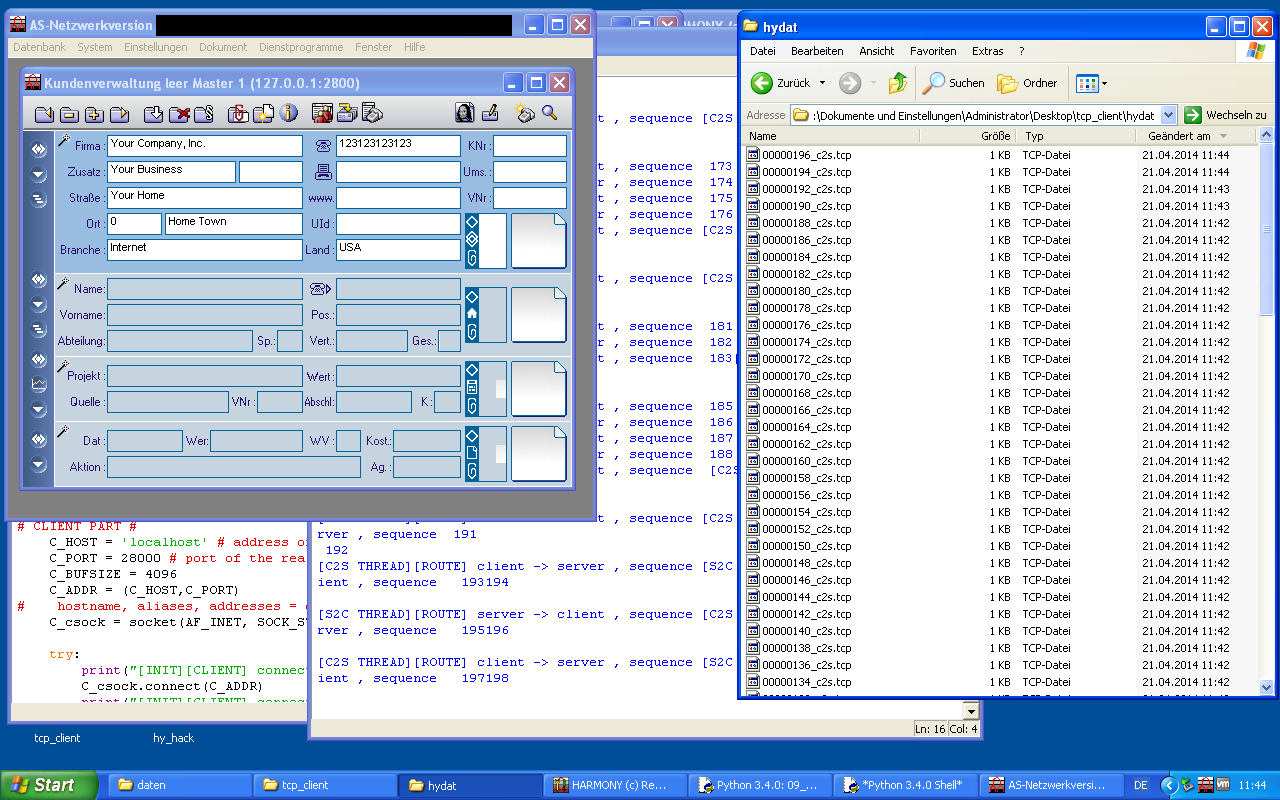} 
\caption[Screenshot of running $M$ and $H$]{This is a screenshot of the Desktop of the machine that is hosting $H$ and $M$ (09\_tcp\_hyrouter.py). At his point the connections between $H_c$ and $M$ as well as $M$ and $H_s$ are established. $M$ is capturing all TCP/IP data packages, which are sent from $H_c$ to $H_s$.} 
\label{fig:06.png} 
\end{figure}

SSH tunneling is activated for remote data exchange only.
In fact, because of missing encryption it is easy to interpose $M$ and run a Man-in-the-middle attack.

The client server protocol of $H$ also misses handshake based authentification procedure that uses cryptographic functions.
It is only string based, meaning to make $M$ establish a connection to the server we just easily have to sniff -- e.g. using wireshark -- the first data packages to get the string, which consists of user name and password. The string can establish communication to the database server without even logging into the system. The system looks up the user in its user list and verifys the password to grant access.

\subsection{Extraction of Raw Data}
The reverse engineering process is easy since $M$ is implemented and runs interpoled between $H_c$ and $H_s$. Once we are connected to the server via $M$, $M$ can log network transmission traffic, so client commands, which gets sent to the server, can be identified.
This process is very easy to handle, since on each pressed button of the {\em GUI} of $H_c$ the client sends the command containing data packages over the established TCP/IP connection.

Obviously $M$ can also be used as standalone client without being connected by client, meaning formally ${ M(\emptyset,D\prime)}$ with a generator function, which generates fake packages from captured data, or ${ M(D_f,D\prime)}$ with $D_f$ is the set of fake data packages, which contains server and database control commands, and $D\prime$ consisting of data packages from $H_s$. To get a proper $D_F$ we have to reverse engineer the full native protocol of $H$.

\begin{figure}
 \centering 
\includegraphics[width=1.0\textwidth]{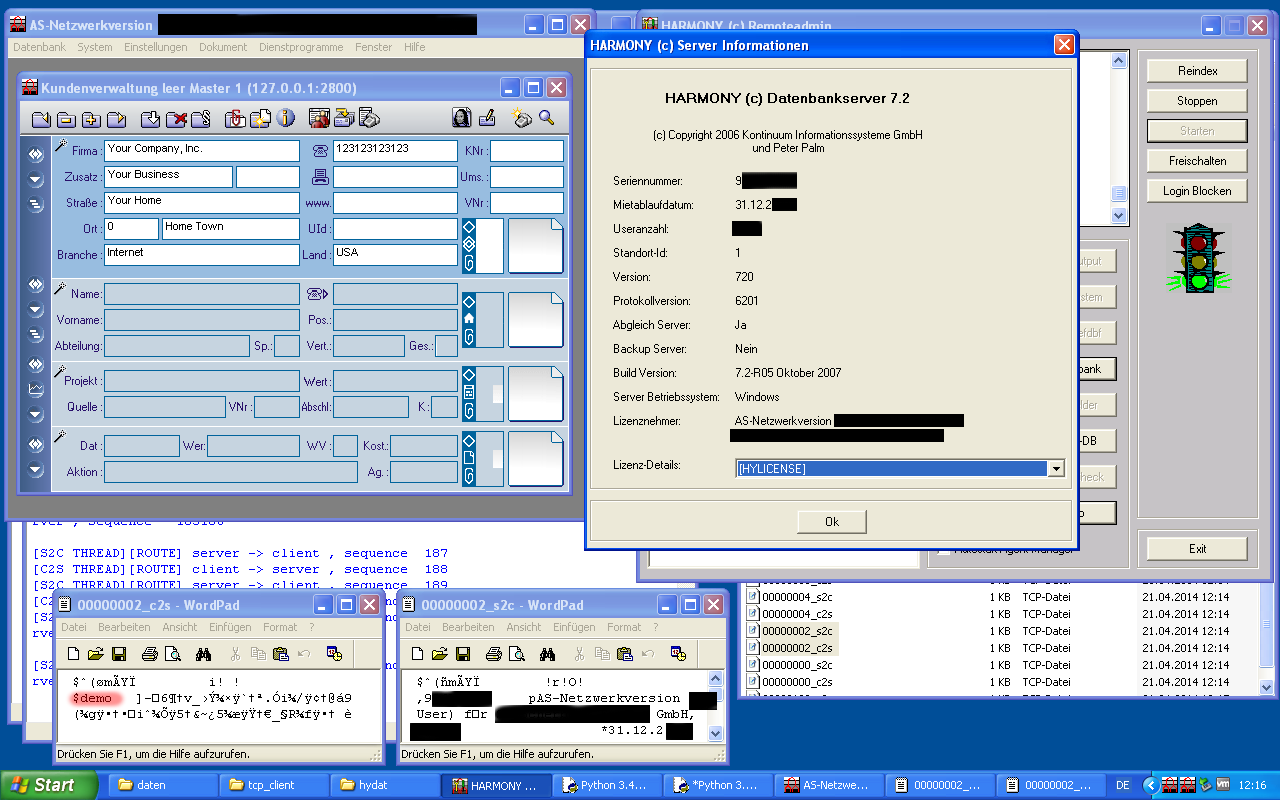} 
\caption[Screenshot of captured TCP/IP data packages]{Each transfered TCP/IP package has been captured by $M$ and saved to the hard disk (see right corner). The user name {\em demo} and the licence information have been transfered in plain text and can be captured.} 
\label{fig:28.png} 
\end{figure}

 This can be done by using $H_c$ on $M$ to connect to $H_s$ as described above.

To keep work time as low as possible, we skip reverse engineering of the full proprietary protocol of $H$. Instead, we just take the most important client commands that transfer data to the client to display it to the user. The traffic contains the raw data, we are looking for. $M$ will copy and save the raw data on-the-fly.


\subsection{Data Cleansing}

Finally, the collected raw data $R$ is full of non-data relevant symbols, so {\em Data Cleansing} (dc) is required for proper data extraction. 
Therefore we have to identify the introduction and the termination bytes of the data. We have to define a structure that will be build by the introduction and termination bytes, which $R$ contains. For our algorithm (see Figure \ref{algo}) we define a function $\mbox{build}_{structure}(\cdot,\cdot)$ that  builds the structure of the extracted data and removes given introduction and termination bytes. The structure of the extracted data is usually based on tree-like structures and/or SQL table schemes, when migrating data to SQL database. On the other hand it is also possible to extract the data into the file system, meaning to store the text data in text files and binary data like pictures as binary image files. So the design of the structure depends on the application.
The implementation of $M$, which is based on the prototype {\em 09\_tcp\_hyrouter.py} will not extract data. 09\_tcp\_hyrouter.py is suposed to capture network traffic only to support reverse engineering process of the trasfer protocol of $H$ in this paper.

\begin{figure}[t]
\fbox{\hspace*{.4cm}
\begin{minipage}{11.7cm}
{\bf Simple Algorithm for Data Cleansing} \\ \\
{\bf Input:} Set $C$ of m known data bytes $C= \{c_0, ... , c_{m-1} \}$ \\
\hspace*{1.22cm}Set of n raw data bytes $R= \{r_0, ... , r_{n-1} \}$\\
{\bf Output:} Clean and structured data in bytes as a set $S_{XML}$\\ \\
{\bf init} empty set $S_{XML}$;\\
{\bf do} $S_{XML} =\mbox{ build}_{structure}(R,C)$\\
{\bf return: $S_{XML}$}  \\
\end{minipage}
}
\caption[Simple Algorithm For Data Cleansing]{Assuming that we know the bytes, which introduce and terminate data, we manage them in a set called $C$ with $C\subset R$. We can use this algorithm for data cleansing. Function $\mbox{build}_{structure}(\cdot,\cdot)$ constructs the data structure by setup filtering $C$ and returning $R\backslash C$ in XML syntax as $S_{XML}$. } 
\label{algo}
\end{figure}

\begin{figure}
 \centering 
\includegraphics[width=0.75\textwidth]{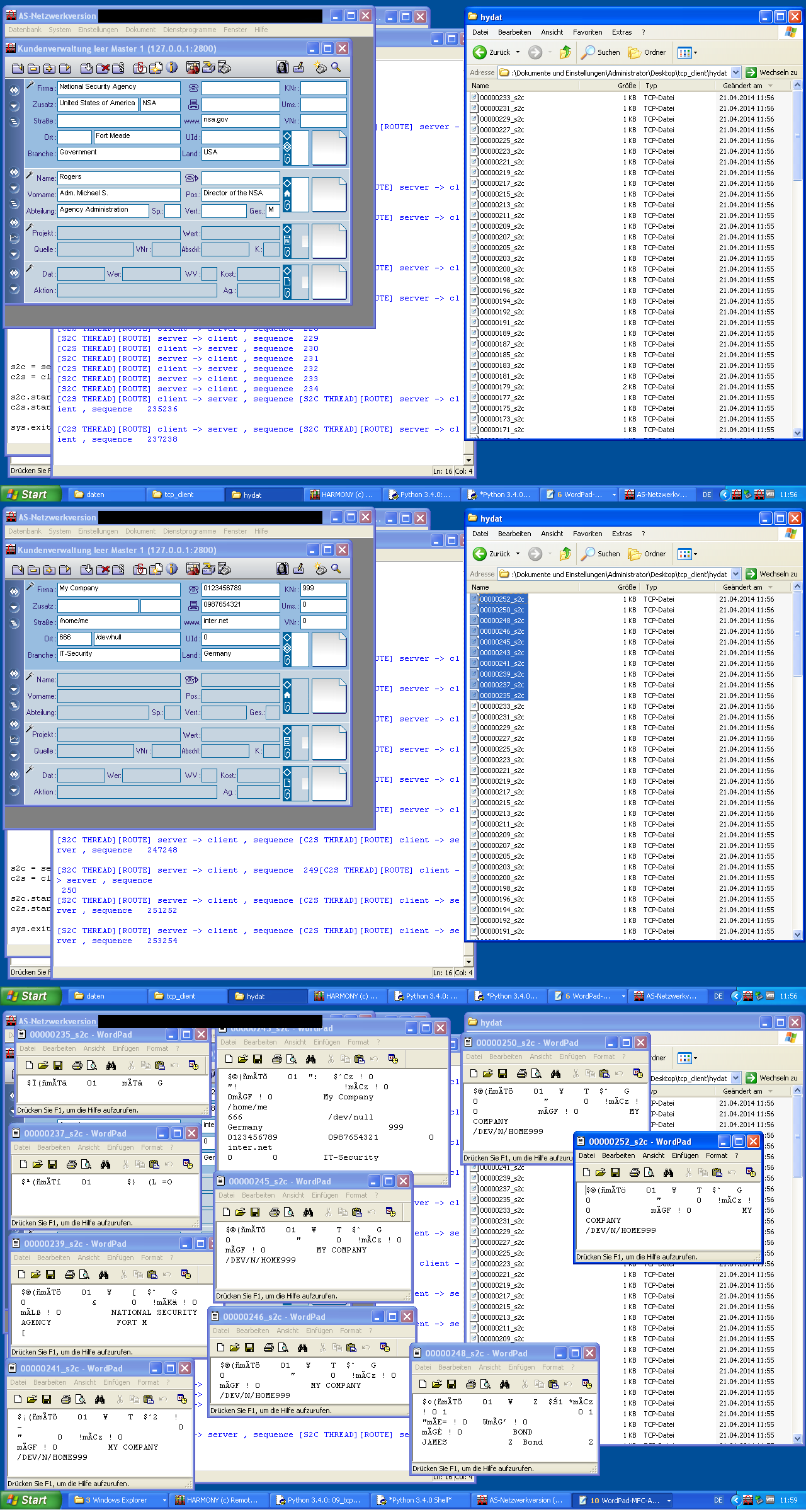} 
\caption[Screenshot of captured TCP/IP data packages]{Each transfered TCP/IP package has been captured and stored to local the hard disk drive. Data was transfered in plain text.} 
\label{fig:23+21+22.png} 
\end{figure}

\begin{koro}
Simple Algorithm for Data Cleansing needs $O(n^2)$ steps to terminate.
\end{koro}
\begin{proof}
We have two sets $R$ and $C$ with $|C|<|R|$, because of $C\subset R$, so initiation is
$$|C|\cdot|R|=m\cdot n = O(n^2). $$ 
Obviously structural analysis takes $O(n^2)$, because we have max. of $m\cdot n$ steps.
Further more, constructing XML based tree structure from data takes $O(n^2)$, because the function has to check input and refer to predefined actions regarding to introduction and termination bytes, so the procedure is powered by n in the length of the input, meaning $O(n^2)$.
Simple Algorithm for Data Cleansing is terminating, because $C$ and $R$ are not infinite. It runs correct, because all byte coded string symbols are predefined and known to the algorithm.
Finally, we get
$$3\cdot O(n^2) = O(n^2)$$
by definition.
\end{proof}

\section{Practice and Future Work}
A demo version of $M$ a.k.a. 09\_tcp\_hyrouter.py written in C/C++ is about TBA for scientific purposes.\footnote{Please refer to cs.burdon.de/downloads for further information.}
A detailed description of the transfer protocol of $H$ is available in the full report of {\em How to extract data from proprietary software database systems using TCP/IP}. The full reverse engineered description of the transfer protocol of $H$ is TBA.
Regarding demonstrative details, please refer to Figure, \ref{fig:06.png}, \ref{fig:28.png} and \ref{fig:23+21+22.png}. 
For detailed information about this topic, please refer to the full report of {\em How to extract data from proprietary software database systems using TCP/IP}.\footnote{Please refer to cs.burdon.de/hy}

$M$ will be redesigned and implemented in C/C++ with additional features like data cleansing function and compatibility to XML. It will automatically convert the output in XML, so data migration to SQL will be possible. Additionally there will be a data extraction mode, which will optionally store the raw or fully extracted data on the hard disk drive -- meaning a data cleansing procedure will be executed for full data extraction. The implemantation will be available for Windows$^{\mbox{\tiny\textregistered}}$, Linux$^{\mbox{\tiny\textregistered}}$, MacOSX$^{\mbox{\tiny\textregistered}}$ and NetBSD. 

\nocite{*}
\bibliographystyle{IEEEtran}
\bibliography{IEEEabrv,IEEEexample}

\begin{thebibliography}{1}
\providecommand{\url}[1]{#1}
\csname url@samestyle\endcsname
\providecommand{\newblock}{\relax}
\providecommand{\bibinfo}[2]{#2}
\providecommand{\BIBentrySTDinterwordspacing}{\spaceskip=0pt\relax}
\providecommand{\BIBentryALTinterwordstretchfactor}{4}
\providecommand{\BIBentryALTinterwordspacing}{\spaceskip=\fontdimen2\font plus
\BIBentryALTinterwordstretchfactor\fontdimen3\font minus
  \fontdimen4\font\relax}
\providecommand{\BIBforeignlanguage}[2]{{%
\expandafter\ifx\csname l@#1\endcsname\relax
\typeout{** WARNING: IEEEtran.bst: No hyphenation pattern has been}%
\typeout{** loaded for the language `#1'. Using the pattern for}%
\typeout{** the default language instead.}%
\else
\language=\csname l@#1\endcsname
\fi
#2}}
\providecommand{\BIBdecl}{\relax}
\BIBdecl

\bibitem{EML}
H.~K. Arndt and O.~Günther. (2000) Environmental {M}arkup {L}anguage ({EML}):
  {W}orkshop 1. Metropolis.

\bibitem{CRM}
W.~Schwetz. (2001) \BIBforeignlanguage{german}{Customer {R}elationship
  {M}anagement. {M}it dem richtigen {CRM}-{S}ystem {K}undenbeziehungen
  erfolgreich gestalten}. Gabler Verlag.

\bibitem{C[13]}
D.~E. Comer. (2013) Internetworking with {TCP/IP}, {V}olume {O}ne - 6th
  {E}dition. Prentice Hall, Inc.

\bibitem{E[05]}
E.~Eilam. (2005) Reversing: {S}ecrets of {R}everse {E}ngineering. Wiley
  Publishing, Inc.

\bibitem{XML}
E.~R. Harold and W.~S. Means. (2004) {XML} in a {N}utshell, 3rd edition.
  O'Reilly.

\bibitem{ATT}
R.~Schifreen. (2006) Defeating the {H}acker: A non-technical {G}uide to
  {C}omputer {S}ecurity. John Wiley and Sons.

\bibitem{ssh3}
\BIBentryALTinterwordspacing
Internet Engineering Task Force. Network Working Group. [Online]. Available:
  \url{http://tools.ietf.org/html/rfc4253}
\BIBentrySTDinterwordspacing

\bibitem{ssh0}
\BIBentryALTinterwordspacing
Internet Engineering Task Force. Network Working Group. [Online]. Available:
  \url{http://tools.ietf.org/html/rfc6668}
\BIBentrySTDinterwordspacing

\end{thebibliography}

\end{document}